\def\DRAFT{}
\documentclass[authoryear,a4paper, 12pt]{elsarticle}
\usepackage{natbib}
\usepackage{amsthm}
\usepackage{amssymb}
\usepackage{graphicx,amsmath}
\usepackage{tikz}
\usepackage{mathtools}
\usepackage{natbib}
\usepackage{booktabs,xcolor}
\usepackage{graphics}
\usepackage{enumitem}
\usepackage{subcaption}
\usepackage{appendix}
\usepackage{adjustbox}
\usetikzlibrary{arrows}
\usetikzlibrary{shapes}

\usepackage{caption}
\usepackage{setspace}
\usepackage[colorlinks=true,linkcolor=blue, urlcolor=blue, citecolor=blue, breaklinks]{hyperref}
\journal{ }

\newtheorem{theorem}{Theorem}

\newtheorem{lemma}{Lemma}

\theoremstyle{definition}

\newtheorem{remark}{Remark}

\usepackage{etoolbox}
\makeatletter
\patchcmd{\ps@pprintTitle}
{Preprint submitted to}
{}
{}{}
\makeatother
\patchcmd{\emailauthor}{(#2)}{}{}{}
\patchcmd{\urlauthor}{(#2)}{}{}{}

\usepackage{tikz}

\newcommand{\ML}{M^{ll}} 
 
\ifdefined\DRAFT
\newcommand{\erel}[1]{\textcolor{green}{(\textbf{Erel:} #1)}}
\newcommand{\er}[1]{\textcolor{green}{#1}}
\newcommand{\josue}[1]{\textcolor{red}{(\textbf{Josue:} #1)}}

\else
\newcommand{\erel}[1]{}
\newcommand{\er}[1]{#1}
\newcommand{\josue}[1]{}

\fi

\begin{document}

\begin{frontmatter}
\title{Obvious Manipulations in Cake-Cutting}

\author[add2,add3]{Josu\'e Ortega\footnote{We are indebted to Thayer Morril, Herv\'e Moulin and Gabriel Ziegler for their helpful comments and to Sarah Fox, Jacopo Gambato and Fabian Spühler for proof-reading this paper. Any errors are our own.}}
\ead{j.ortega@qub.ac.uk}
\author[add4]{Erel Segal-Halevi}
\ead{erelsgl@gmail.com.}
\address[add2]{Queen's Management School, Queen's University Belfast, UK.}
\address[add3]{Center for European Economic Research, Mannheim, Germany.}
\address[add4]{Department of Computer Science, Ariel University, Israel.}

\begin{abstract}
In cake-cutting, strategy-proofness is a very costly requirement in terms of fairness: for $n=2$ it implies a dictatorial allocation, whereas for $n\geq 3$ it requires that one agent receives no cake. We show that a weaker version of this property recently suggested by Troyan and Morril, called {\it not-obvious manipulability}, is compatible with the strong fairness property of {\it proportionality}, which guarantees that each agent receives $1/n$ of the cake. Both properties are satisfied by the {\it leftmost leaves} mechanism, an adaptation of the Dubins - Spanier moving knife procedure. Most other classical proportional mechanisms in literature are obviously manipulable, including the original moving knife mechanism. Not-obvious manipulability explains why leftmost leaves is manipulated less often in practice than other proportional mechanisms.
\end{abstract}

\end{frontmatter}
	
	\newpage
	\setcounter{footnote}{0}

\section{Introduction}
\label{sec:introduction}
The division of a single good among several agents who value different parts of it distinctly is one of the oldest fair division problems, going as far back as the division of land between Abram and Lot (Genesis 13). Since its formalization as the cake-cutting problem \citep{steinhaus1948}, this research question has inspired a large interdisciplinary literature which has proposed mechanisms that produce fair allocations without giving agents incentives to misrepresent their preferences over the cake. Unfortunately, \cite{branzei2015dictatorship} show that any deterministic strategy-proof mechanism suggests a dictatorial allocation for $n=2$, and that one agent must receive no cake at all for $n\geq 3$, documenting a strong tension between fairness and incentives properties.

Nevertheless, recent results in applied mechanism design have shown that, even if mechanisms can be manipulated in theory, they are not always manipulated in practice. Some manipulations are more likely to be observed than others, particularly those which are salient or require less computation. Based on this observation, \cite{troyan2019obvious} have proposed a weaker version of strategy-proofness for direct mechanisms, called {\it not-obvious manipulability} (NOM). They define a manipulation as obvious if it yields a higher utility than truth-telling in either the best- or worst-case scenarios. A mechanism is NOM if it admits no obvious manipulation. Their notion of NOM is a compelling one, since it does not require prior beliefs about other agents' types, and compares mechanisms only based on two scenarios which are particularly salient and which require less cognitive effort to compute. They show that NOM accurately predicts the level of manipulability that different mechanisms experience in practice in school choice and auctions. 

In this paper, we provide a natural extension of NOM to indirect mechanisms, and show that the stark conflict between fairness and truth-telling in cake-cutting disappears if we weaken strategy-proofness to NOM. In particular, NOM is compatible with the strong fairness property of {\it proportionality}, which guarantees each agent $1/n$ of the cake. Both properties are satisfied by a discrete adaptation of the moving knife mechanism \citep{dubins1961}, in which all agents cut the cake and the agent with the smallest cut receives all the cake to the left of his cut and leaves. This procedure is also procedurally fair and easy to implement in practice. NOM is violated by most other classical proportional mechanisms, even by the original Dubins-Spanier procedure, which shows that theoretically equivalent mechanisms may have different ``obvious'' incentive properties for boundedly rational agents. NOM partially explains why leftmost leaves is manipulated less frequently than other cake-cutting mechanisms in practice \citep{kyropoulou2019fair}.

\section{Related Literature}

The cake-cutting problem has been studied for decades, given its numerous applications to the division of land, inheritances, and cloud computing \citep{brams1996,moulin2004}. Most of the cake-cutting literature studies indirect revelation mechanisms, in which agents are asked to reveal their valuation over the cake via specific messages such as cuts and evaluations. In particular, the computer science literature focuses on the so called \emph{Robertson--Webb mechanisms} \citep{robertson1998}, in which agents can only use two types of messages: either an agent \emph{cuts} a piece of the cake having a specific value, or he \emph{evaluates} an existing piece by revealing his utility for it. Most well-known mechanisms in the literature, such as cut and choose, can be expressed as a combination of these two operations.

The main result relating fairness and incentive properties in cake-cutting was put forward by \cite{branzei2015dictatorship}. They show that with three or more agents, every strategy-proof Robertson--Webb mechanism assigns no cake to at least one agent. Their result builds on a weaker result by \cite{kurokawa2013cut} regarding mechanisms that have a bounded number of messages. 

To allow for some degree of fairness and truth-telling, the literature has considered four research avenues. 

The first one is to use \emph{randomized mechanisms}. For example, \cite{mossel2010} show that truth-telling in expectation and approximate proportionality can be obtained with probabilistic mechanisms. In this paper we consider deterministic mechanisms only.

The second is to restrict the set of possible valuations over the cake. In this line, \cite{chen2013} provide a complex deterministic mechanism that is both strategy-proof and proportional for a restricted class of utilities called \emph{piecewise linear}. Their mechanism may waste pieces of cake which remain unassigned, something that never occurs with the proportional and NOM mechanism we identify. 

A third research avenue studies allocation mechanisms in which an agent can only increase his utility by $\epsilon$ by misrepresenting his preferences, compared to the utility he obtains from truth-telling. \cite{menon2017deterministic} and \cite{kyropoulou2019fair} obtain bounds on the amount of extra utility that agents can guarantee by lying in proportional cake-cutting mechanisms.


In this paper, we explore a fourth way to escape the conflict between fairness and truth-telling, which is relaxing the strategy-proofness definition to the one of NOM. NOM is a stronger version of an existing weak truth-telling property in the cake-cutting literature called \emph{maximin strategy-proofness}, proposed by \cite{brams2006,brams2008proportional}. We discuss the relationship between these two concepts in detail in the next section. 

\section{Model}
\subsection{Outcomes, agents and types}
In a general mechanism design problem, there is a set of possible outcomes $\mathcal X$, a set of agents $N$, and a set of possible agent-types $\Theta= \Pi_{i \in N} \Theta_i$. 
In the particular case of cake cutting, there is an interval $[0,1]$ called the \emph{cake}. A union of subintervals of $[0,1]$ is called a \emph{piece} of the cake.
The outcome-set $\mathcal X$ is the set of \emph{allocations} --- ordered partitions of $[0,1]$ into $n$ disjoint pieces. 
$X\in \mathcal X$ denotes an arbitrary allocation and $X_i$ denotes the piece allocated to agent $i$. 

The type-set $\Theta$ is the set of possible \emph{utility functions}, where a utility function is a non-atomic measure --- an additive function that maps each piece to a non-negative number.
The utility of the entire cake is normalized to $1$. 
The utility of agent $i$ with type $\theta_i$ is denoted by $u_i$. That $u_i$ is non-atomic allows us to ignore the boundaries of intervals. Another implication of nonatomicity is that $u_i$ is divisible, i.e. for every subinterval $[x, y]$ and $0\leq \lambda \leq 1$, there exists a point $z \in [x, y]$ such that $u_i([x, z];\theta_{i}) = \lambda \, u_i([x, y];\theta_{i})$. \footnote{Both $u_i$ and $\theta_i$ are equivalent --- the type of an agent is the agent's utility, so our notation is slightly redundant. Nevertheless, we use it to make the comparison with \cite{troyan2019obvious} straightforward.}

\subsection{Extensive forms and mechanisms}
An \emph{extensive form} is an arborescence $A$ that consists of a set of labelled nodes $H$ and a set of directed edges $E$.\footnote{An arborescence is a directed, rooted tree in which all edges point away from the root.} The root node is $h_0$. Each terminal node is labelled with an allocation of the cake. Each non-terminal node $h$ is labelled with a non-empty subset of agents $N(h)$ who have to answer a query about their type. 
$N(h)$ is said to be the set of players \emph{active at $h$}. 

In a general extensive form, the query may be arbitrary, for example asking agents to  fully reveal their type. 
In a \emph{Robertson--Webb extensive form}, only two types of queries are allowed:
\begin{enumerate}
	\item \emph{Cut query}: the query cut$(i;x,\alpha)$ asks agent $i$ the minimum point $y \in [0; 1]$ such that $u_i([x,y];\theta_i)=\alpha$; where $\alpha \in [0,1]$ and $x$ is an existing cut point or 0 or 1. The point $y$ becomes a cut point.
	\item \emph{Eval query}: 
	the query eval$(i; x, y)$
	asks agent $i$ for its value for the interval $[x, y]$, that is, eval$(i; x, y) = u_i([x, y];\theta_i)$ where $x,y$ are a previously made cut points or 0 or 1.
\end{enumerate}	

All agents must reply to the queries and their answers must be dynamically consistent, meaning that there must be a possible type for which such a sequence of answers is truthful.\footnote{This is a standard requirement \citep{branzei2015dictatorship}. For example, if agent $i$ is asked the query eval$(i; 0.3, 1)$ and replies a value of $0$, then if later asked the query cut$(i;0,0.5)$ his answer must be in the interval $[0,0.3)$.} However, agents' answers may be untruthful. For every possible combination of agents' answers to the queries at node $h$, there is an edge from $h$ to another node $h'$.
Thus, answers to queries and edges have a one-to-one relationship. Note that, since the valuations are real numbers, there are uncountably many edges emanating from each node.
Figure \ref{fig:extensiveform} illustrates a small subset of an extensive form $A$.

\begin{figure}[ht]
\centering
\begin{center}
	\begin{tikzpicture}
	\tikzstyle{every node}=[draw=black, ellipse, minimum width=100pt,align=center]
	\node (h0) at (0,1) {cut$(1;0,0.5)$};
	\node[circle,draw=white, fill=white,   inner sep=0pt,minimum size=9pt] (h1) at (-2.5,-2) {};
	\node[circle,draw=white, fill=white,   inner sep=0pt,minimum size=9pt] (h11) at (-4,-2) {};
	\node (h2) at (.5,-2) {eval$(2; 0, 0.6)$};
	\node[circle,draw=white, fill=white,   inner sep=0pt,minimum size=9pt] (h3) at (2.7,-2) {};
	\node[circle,draw=white, fill=white,   inner sep=0pt,minimum size=9pt] (h33) at (4.2,-2) {};
	\node (h7) at (-3.5,-4) {$[0,0.6),[0.6,1]$};
	\node (h8) at (3.5,-4) {$[0.6,1],[0,0.6)$};
\path[every node/.style={sloped,anchor=south,auto=false}]
(h0) edge              node {} (h1)            
(h0) edge              node {0.01} (h11)            
(h0) edge              node {0.6} (h2)
(h0) edge              node {0.9} (h33)
(h0) edge              node {0.7} (h3)
(h2) edge              node {0.2 } (h7)
(h2) edge              node {0.8} (h8);

\draw[->] (h0) -- (h1);
\draw[->] (h0) -- (h11);
\draw[->] (h0) -- (h2);
\draw[->] (h0) -- (h3);
\draw[->] (h0) -- (h33);
\draw[->] (h2) -- (h7);
\draw[->] (h2) -- (h8);

\draw[dotted,line width=0.5mm] (-1.2,-1)--(0,-1);
\draw[dotted,line width=0.5mm] (.5,-1)--(1.7,-1);
\draw[dotted,line width=0.5mm] (-1.1,-3.3)--(1.4,-3.3);
	\end{tikzpicture}
	\caption{An example of an extensive form representing cut and choose.}
	\label{fig:extensiveform}
\end{center}
\end{figure}
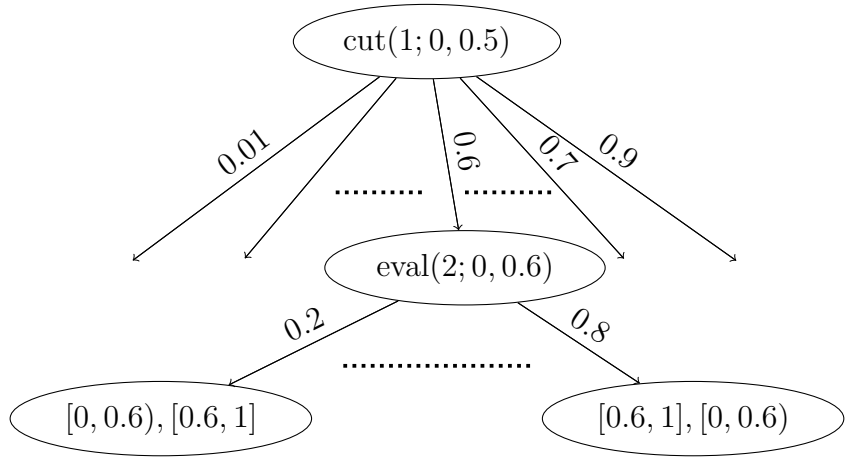

We make the following two standard informational assumptions \citep{moore1988subgame}. First, at each node $h$ all agents know the entire history of the play. Second, if more than one agent is active at node $h$, they answer their corresponding queries simultaneously. 

The indirect mechanism $M$ corresponding to extensive form $A$ takes as input the answers to each query in $A$ and returns the allocation obtained at the corresponding terminal node. This is, the input to the mechanism consists of a path from the root node to a terminal node labelled with an allocation. At node $h$, the edge corresponding to a truthful answer by all agents in $N(h)$ with a type profile $\theta$ is denoted by $e^h(\theta)$, and $E(\theta)=\{e \in E : e^h(\theta) \text{ for some } h \in H\}$. This is, $E(\theta)$ is the set of all edges corresponding to truthful reports. The set $E(\theta)$ is an edge cover, i.e is a set of edges incident to every node in $A$.\footnote{Edge covers represent a complete contingent plan of answers.} An alternative edge cover in which all agents except $i$ answer all queries truthfully, and agent $i$ answers the queries as if he was of type $\theta_i'$, is denoted by $E(\theta_i',\theta_{-i})$. By our assumption of dynamically consistent answers, each possible set of untruthful answers by each agent is associated to a possible type in $\Theta_i$. $M_i(E(\theta_i,\theta_{-i}))$ denotes agent $i$'s individual allocation in mechanism $M$ when agents' answers correspond to the set of edges $E(\theta_i,\theta_{-i})$. 
A special case of our definition is a direct revelation mechanism, when there is a unique non-terminal node --- the root $h_0$ ---, $N(h_0)=N$ and the query for each agent is to fully reveal his type.

A mechanism $M$ is called {\it proportional} if it guarantees to a truthful agent a utility of $1/n$ regardless of what every other agent reports, i.e. $u_i(M_i(E(\theta_i,\theta'_{-i});\theta_i))\geq 1/n$ for all $i \in N$, for all $\theta_i \in \Theta_i$, and all $\theta'_{-i} \in \Theta_{-i}$. 

\subsection{Manipulations}

A \emph{manipulation} by agent $i$ is a set of (partially) untruthful answers that satisfy the consistency requirement, i.e., correspond to some type $\theta_i' \neq \theta_i$.

A manipulation is called \emph{profitable} (for mechanism $M$ and type $\theta_i$)
if $u_i(M_i(E(\theta'_i,\theta_{-i}));\theta_i) > u_i(M_i(E(\theta_i,\theta_{-i}));\theta_i)$. If some type of some agent has a profitable manipulation, then mechanism $M$ is called \emph{manipulable}.

A mechanism is called {\it strategy-proof} (SP) if it is not manipulable, that is,  $u_i(M_i(E(\theta_i,\theta_{-i}));\theta_i) \geq u_i(M_i(E(\theta'_i,\theta_{-i}));\theta_i)$ for all $i \in N$, for all $\theta_i, \theta'_i \in \Theta_i$, and all $\theta_{-i} \in \Theta_{-i}$.  
This definition of SP is very demanding: it requires a truthful report from \emph{every} agent, \emph{every} time he is asked to answer a query.\footnote{Strategy-proofness is a notion more commonly used for direct-revelation mechanisms. This extension to indirect mechanisms follows the definition of \citet{kurokawa2013cut}.} Even if just one type of an agent has an incentive to give a non-truthful answer to a single query, the mechanism is no longer strategy-proof. 

\citet{troyan2019obvious} suggest a weaker version of SP for direct-revelation mechanisms, which only compares the best and worst case scenarios from both truthful and untruthful behavior. We extend their definition to indirect mechanisms as follows. A mechanism $M$ is {\it not-obviously manipulable} (NOM) if, for any profitable manipulation 
of agent $i$, corresponding to pretending being a type $\theta_i'$, the following two conditions hold:\footnote{They present this definition using maximum and mininum, which may not exists with a continuous cake; instead we consider the supremum and infimum.}
\begin{eqnarray}
\label{eq:nom-inf}
\inf_{\theta_{-i}} u_i(M_i(E(\theta'_i,\theta_{-i}));\theta_i) \leq \inf_{\theta_{-i}} u_i(M_i(E(\theta_i,\theta_{-i}));\theta_i)
\\
\label{eq:nom-sup}
\sup_{\theta_{-i}} u_i(M_i(E(\theta'_i,\theta_{-i}));\theta_i) \leq \sup_{\theta_{-i}} u_i(M_i(E(\theta_i,\theta_{-i}));\theta_i)
\end{eqnarray}
	
If any of the previous two conditions do not hold for some $\theta'_i$, then $\theta'_i$ is said to be an \emph{obvious manipulation} for agent $i$ with type $\theta_i$; and the mechanism $M$ is \emph{obviously manipulable}.
In other words, a manipulation is obvious if it either makes the agent better off in the worst-case, or if it makes him better off in the best-case. 
This is a strengthening of the \emph{maximin strategy-proofness} defined by 
\cite{brams2006,brams2008proportional}, who only impose condition (1). Brams, Jones and Klamler write: {\it ``We assume that players try
to maximize the minimum-value pieces (maximin pieces) that they can guarantee for
		themselves, regardless of what the other players do. In this sense, the players are
		risk-averse and never strategically announce false measures if it does not guarantee
		them more-valued pieces"}.
	
Both relaxations of strategy-proofness have the advantages that agents do not require beliefs about other agents' actions, and that comparing best- and worst-cases scenarios requires less cognitive effort than comparing expected values using an arbitrary distribution over agents' types. However, maximin strategy-proofness is a mild property that is satisfied by a very large class of mechanisms.\footnote{\cite{chen2013} call maximin SP a ``strikingly weak notion of truthfulness''.} On the other hand, NOM is a property that most classical proportional mechanisms in the literature fail, with the {\it leftmost leaves} mechanism being a remarkable exception (Theorem \ref{thm:leftmost-leaves}). 


Before we proceed to present some examples, we clarify a difference between the notion of indirect mechanisms that we use (standard in the computer science literature) and the definition in economics \citep{moore1988subgame}. In economics, each non-terminal node is labelled not with a query, but with an action. Thus, to define leftmost leaves one could just specify that, at each period $t$, an agent's possible actions are to cut a cake at any point, and the one who provides the smallest cut exits with the leftmost part. There is no question where to cut: agents simply cut the cake wherever they want. But, if one does not specify the query asked to each player, it is not {at all} clear what truth-telling behavior is. Where should an agent with uniform valuation over the cake cut when dividing a cake against 3 agents if he is not asked a specific question? 
The computer science definition emphasizes that the mechanism designer can make a mechanism hard to manipulate by cleverly choosing which queries to ask at each node in the extensive form. Indeed, if we slightly modify the queries asked in \emph{leftmost leaves}, the mechanism becomes obviously manipulable, as we show after the proof of Theorem \ref{thm:leftmost-leaves}.

\section{Results}
We first show, using an example, that many classic cake-cutting procedures are obviously manipulable.

\begin{theorem}
\label{thm:om}
The cake-cutting mechanisms known as cut-and-choose, cut-middle, Banach-Knaster last diminisher and Dubins-Spanier moving knife are all obviously manipulable.
\end{theorem}
\begin{proof}
For simplicity, consider a cake-cutting problem with piecewise uniform valuations, i.e. agents either like or dislike certain intervals, and each desirable interval of the same length has the same value. 
One agent, called Blue, has valuations as  in Figure \ref{fig:example1}.

\begin{figure}[h]
	\centering
		\begin{tikzpicture}

			\fill[blue] (0,0) rectangle (1,1);
			\fill[blue] (9,0) rectangle (10,1);
			\fill[blue] (4,0) rectangle (6,1);
			\draw[step=1cm,gray,very thin] (0,0) grid (10,1);
			\foreach \x in {1,2,3,4,5,6,7,8,9}
			\draw (\x cm,1pt) -- (\x cm,1pt) node[anchor=north] {$.\x$};
			\end{tikzpicture}
			\caption{The preferences of a (blue) agent over the cake.} \label{fig:example1}
		\end{figure}
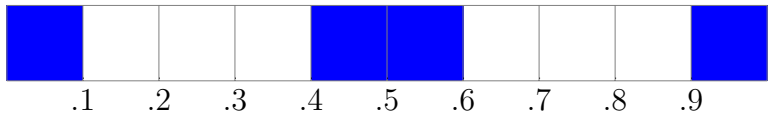

\begin{figure}[h]
	\centering
	\begin{subfigure}{.45\linewidth}
		\centering
		\begin{tikzpicture}[scale=.55]
		\fill[blue] (0,.5) rectangle (1,1);
		\fill[blue] (9,.5) rectangle (10,1);
		\fill[blue] (4,.5) rectangle (6,1);
		\fill[red] (0,0) rectangle (4,.5);
		
		\draw[->, ultra thick] (4,2) -- (4,1.2);
		\draw[step=1cm,gray,very thin] (0,0) grid (10,1);
		\draw[green, line width=0.1cm] (4,0) rectangle (10,1);
		\foreach \x in {1,2,3,4,5,6,7,8,9}
		\draw (\x cm,1pt) -- (\x cm,1pt) node[anchor=north] {$.\x$};
		\end{tikzpicture}
		\caption{Cut-and-choose.} \label{fig:M1}
	\end{subfigure}%
	\begin{subfigure}{.45\linewidth}
		\centering
		\begin{tikzpicture}[scale=.55]
		\fill[blue] (0,.5) rectangle (1,1);
		\fill[blue] (9,.5) rectangle (10,1);
		\fill[blue] (4,.5) rectangle (6,1);
		\draw[step=1cm,gray,very thin] (0,0) grid (10,1);
		
		\fill[red] (0,0) rectangle (.5,.5);
		\draw[green, line width=0.1cm] (.625,0) rectangle (10,1);
		\draw[->, ultra thick] (1,2) -- (1,1.2);
		
		\foreach \x in {1,2,3,4,5,6,7,8,9}
		\draw (\x cm,1pt) -- (\x cm,1pt) node[anchor=north] {$.\x$};
		\end{tikzpicture}
		\caption{Cut-middle.} \label{fig:M2}
	\end{subfigure}
	\caption{Obvious manipulations for agent blue (black arrow), supporting preferences for the other agent (red), final allocation received by blue agent (green).}
\end{figure}
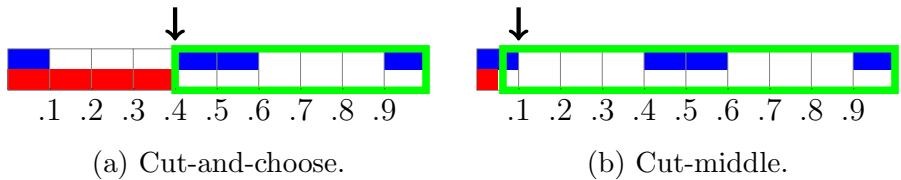

In the well-known {\it cut-and-choose} mechanism, Blue (the cutter) is asked  $\text{cut}(\text{blue};0,1/2)$. Truthful behavior requires him to cut at $0.5$, guaranteeing a utility of $0.5$ in all cases. If Blue chooses a profitable manipulation instead, say to cut at $0.4$, the best case is that the other agent chooses the left piece of the cake, leaving to Blue a utility of $0.75$.
Thus, in inequality \eqref{eq:nom-sup},
the supremum at the left-hand side is at least $0.75$ while the supremum at the right-hand side is $0.5$.
Cut-and-choose is therefore obviously manipulable. 
		
In the {\it cut-middle} mechanism, both agents are asked $\text{cut}(i;0,1/2) = x_i$ simultaneously, and the cake is divided at $\frac{x_1+x_2}{2}$, with each agent obtaining the part of the cake which contains his cut. If Blue is truthful and cuts at $0.5$, the best case is that the other agent cuts at $\epsilon$ (or $1-\epsilon$) and thus the cut point becomes $\frac{0.5+\epsilon}{2}\approx 0.25$, so Blue receives $0.75$ utility. 
Nevertheless, if Blue chooses a profitable manipulation such as cutting at $\epsilon$, the best that could happen is that the other agent cuts at $\eta < \epsilon$, and thus Blue receives a utility almost equal to $1$. 
In inequality \eqref{eq:nom-sup},
the supremum at the left-hand side is $1$ and the supremum at the right-hand side is $0.75$.

We conclude that cut-middle, too, is obviously manipulable.

Cut-and-choose and cut-middle are mechanisms for dividing a cake among two agents. We now turn to mechanisms that can be used to divide cake among two or more agents.

	
First, consider the {\it Banach-Knaster last diminisher} mechanism \citep{steinhaus1948}, in which agents are assigned a fixed order. The first agent is asked the point $x_1=\text{cut}(1;0,1/n)$, 
and is tentatively assigned the piece $[0, x_1]$.
Then, the second agent has an option to ``diminish'' this piece: 
he is asked the point $x_2=\text{cut}(2;0,1/n)$, and 
if $x_2<x_1$, then 
the previous tentative assignment is revoked, and agent 2 is now tentatively assigned the piece $[0, x_2]$.
This goes on up to agent $n$. Then, the tentative assignment becomes final: some agent $k$ receives the piece $[0,x_k]$ and the other agents recursively divide the remaining cake $[x_k,1]$.

This procedure is obviously manipulable even with two agents. 
For example, 
consider 
Figure \ref{fig:example1}(a).
Suppose agent 1 answers truthfully $x_1 = 0.5$.
If agent 2 answers truthfully $x_2 = 0.2$, then he gets a value of exactly $0.5$. 
A profitable manipulation for agent 2 is to answer $x_2' = x_1 - \epsilon$: it yields a utility of $1$.
This is true in both the best- and worst-case scenarios, which are the same in this case.


We turn to the {\it Dubins-Spanier moving knife} \citep{dubins1961}, in which a knife moves from 0 to 1 until the first agent stops it at the point $x$ for which $u_i([0,x];\theta_i)=0.5$, receiving himself the interval $[0,x]$.
The best that can happen to a truthful agent is that the other agent cuts the cake at $\epsilon$, so that he receives all the remaining cake which gives him almost 1 utility. However, once the moving knife has reached point $x$, a truthful agent should stop the knife, implying that he gets $0.5$ with certainty, whereas a profitable manipulation would yield a higher utility in the best-case scenario (one example is to stop the knife at $1-\epsilon$). Thus, Dubins-Spanier moving knife is also obviously manipulable. 
\end{proof}

We now present a slight variation of the Dubins-Spanier moving knife mechanism, in which, in each period, all agents are asked the cut query simultaneously rather than sequentially.
We call this variant \textbf{leftmost leaves}.
The leftmost leaves mechanism works as follows:
\begin{itemize}
\item In period $1$ all agents are asked $\text{cut}(i;0,1/n)$. The agent who cuts the cake at the smallest point (denoted $x^1$) leaves with the interval $[0,x^1]$. 
In case of a tie, the agent with the smallest index of all those who cut at $x^1$ leaves with $[0,x^1]$.
\item In period 2, every remaining agent $i$ is asked  $\text{cut}(i;x^1,\frac{u_i([x^{1},1];\theta_i)}{n-1})$.
The agent who submits the smallest point (denoted $x^2$) leaves with the interval $(x^1,x^2]$. 
In case of a tie, the agent with the smallest index of all those who cut at $x^2$ leaves with $(x^1,x^2]$.

\item 
In period $t$, all remaining agents are asked $\text{cut}(i;x^{t-1},\frac{u_i([x^{t-1},1];\theta_i)}{n-t+1})$. The procedure continues until only one agent remains. That agent receives $(x^{n-1},1]$. 
	\end{itemize}

Despite leftmost leaves being equivalent (in a sense in which we describe below) to the Dubins-Spanier moving knife mechanism, they differ in terms of obvious manipulability.

	
\begin{theorem}
\label{thm:leftmost-leaves}
The leftmost leaves mechanism is proportional and not-obviously-manipulable.
\end{theorem} 

Before presenting the proof, we present a few remarks about this result and its possible extensions. 

First, note that leftmost leaves is theoretically equivalent to Dubins-Spanier moving knife mechanism, in that when applied to truthful agents with the same types, both mechanisms always suggest the same allocation. How can mechanisms that are theoretically equivalent, such as the two we just presented, rank differently in terms of incentives? This idea goes back to \cite{li2017obviously}, who shows that two equivalent mechanisms, such as the ascending auction and the second price auction, in which bidding truthfully is a weakly dominant strategy, are different in terms of incentive properties for boundedly rational agents. The intuition in both results is similar: both in the second price auction and in leftmost leaves, agents have no restriction in the prior about their opponents' types when they reveal their type through either their bids or their cuts; whereas in both the ascending auction and the moving-knife procedure, the fact that the knife or the clock has reached some point tells the agents' something about their opponents' types, and thus modifies what to expect in the best- and worst-case scenarios.

Second, the leftmost leaves mechanism satisfies several other desiderata that make it a good candidate to divide goods in practice. Is not only proportional and NOM, but also procedurally fair (up to tie-breaks) \citep{crawford1977,nicolo2008}, since agents' identities do not affect the allocation produced. It also generates an assignment of a connected piece of cake for each agent, a desirable property for applications such as the division of land \citep{segal2017fair}. 

Third, a follow-up question to Theorem \ref{thm:leftmost-leaves} is whether leftmost leaves is the only proportional and NOM mechanism in cake-cutting. The answer is no, as leftmost leaves can be slightly modified in several ways retaining both properties. One of such modifications is to start cutting the cake from the left instead of from the right.\footnote{\emph{Rightmost leaves} is the only mechanism to divide cake among two agents that is weakly Pareto optimal, proportional and resource monotonic under specific restrictions on agents' utilities \citep{segal2018resource}.} Another less trivial one is a modification of the protocol of \cite{even1984}, which works exactly the same as leftmost leaves for $n=2$ but requires fewer queries for larger values of $n$ (the proof is an almost verbatim copy of the proof of Theorem \ref{thm:leftmost-leaves}, so we omit it).\footnote{The ``leftmost leaves Even-Paz'' mechanism works as follows (for clarity let $n$ be a power of 2). Given a cake $[y,z]$, all agents choose cuts $x_i$ such that $u_i([y,x_i];\theta_i)=u_i([y, z];\theta_i)/2$. Let $x^*$ be the median of the $n$ cuts, rounded to the nearest smallest integer. Then the procedure breaks the cake-cutting problem into two: all agents who choose cuts $x_i \leq x^*$ are to divide the cake $[y,x^*)$, whereas all agents who chose cuts above $x^*$ are to divide the cake $[x^*,z]$.  Each half is divided recursively among the $n/2$ partners assigned to it. When the procedure is called with a singleton set of agents $\{i\}$ and an interval $I$ it assigns $X_i = I$. For example, if $n=4$, agents cut the cake in two equivalent pieces and the cake is cut at the second smallest cut.
Then the two agents with the smallest (largest) cuts play \emph{leftmost leaves} on the left (right) side of the cake.} Obtaining a characterization of all NOM and proportional mechanisms remains an interesting, albeit challenging, open question. 


\section{Proof of Theorem \ref{thm:leftmost-leaves}}

We denote the leftmost leaves mechanism by $\ML$. We show that $\ML$ is not-obviously manipulable. Since $\ML$ is an anonymous mechanism in which the identity of the agents does not play a role, it is necessary to check conditions \eqref{eq:nom-inf} and \eqref{eq:nom-sup} only for one agent, denoted by $i$.


First, we show that no manipulation yields a higher utility in the worst-case scenario.
We use the following lemma.

\begin{lemma}
\label{lem:inf}
 At period $t$,
$$
	\inf_{\theta_{-i}} u_i(\ML(E(\theta_i,\theta_{-i}));\theta_i) = \frac{u_i([x^{t-1},1];\theta_i)}{n-t+1}
$$
\end{lemma}

\begin{proof}
If $x_i^t$ is the smallest cut at period $t$, then the result is immediate. Otherwise,  
	$$
	u_i([x^{t-1},x^t];\theta_i)\leq \frac{u_i([x^{t-1},1];\theta_i)}{n-t+1}
	$$

	A direct implication is that the remainder of the cake $[x^t,1]$ must be worth at least $\frac{n-t}{n-t+1}$ of $u_i((x^{t-1},1);\theta_i)$, i.e.
	$$
	u_i([x^t,1];\theta_i) \geq \frac{n-t}{n-t+1} u_i([x^{t-1},1];\theta_i)
	$$
	
	Dividing both sides by $n-t$
	$$
	\frac{u_i([x^t,1];\theta_i)}{n-t} \geq \frac{ u_i([x^{t-1},1];\theta_i)}{n-t+1}
	$$
	Note that the left-hand side of the previous expression is the utility that the truthful agent $i$ would receive if he cut the cake at the smallest point in period $t+1$. If his cut was not the smallest at period $t+1$, a recursive formulation shows that he would receive a share of the cake that he values even more in period $t+2$. Thus, the worst that can happen to a truthful agent in period $t$ is to obtain a utility of $\frac{u_i([x^{t-1},1];\theta_i)}{n-t+1}$. 
 \end{proof}
 
Setting $t=1$ in Lemma \ref{lem:inf} shows that leftmost leaves is proportional:

$$
	\inf_{\theta_{-i}} u_i(\ML(E(\theta_i,\theta_{-i}));\theta_i) = \frac{u_i([0,1];\theta_i)}{n}
$$

Now we show that any manipulation at period $t$, $\hat x_i^t \neq x_i^t$, yields a utility weakly smaller than $\frac{u_i([x^{t-1},1];\theta_i)}{n-t+1}$ in the worst-case scenario.

If $\hat x_i^t < x_i^t$, in the worst-case scenario $\hat x_i^t$ would be the smallest cut, and agent $i$ would therefore receive the allocation $[x^{t-1},\hat x_i^t]$, which by construction yields a weakly lower utility than the allocation $[x^{t-1},x^t_i]$. 

In the other case, if $\hat x_i^t > x_i^t$, in the worst-case scenario the smallest cut in period $t$, $x^t$, would be such that $x_i^t< x^t<\hat x_i^t$. Thus, the rest of the cake $[x^t,1]$ is such that
$$
u_i([x^t,1];\theta_i) < \frac{n-t}{n-t+1} u_i([x^{t-1},1];\theta_i)
$$

In period $t+1$, agent $i$ cannot choose a manipulation $\hat x_i^{t+1}\leq x_i^{t+1}$, as per the previous step he would receive a smaller utility, so agent $i$ must choose a manipulation $\hat x_i^{t+1}>x_i^{t+1}$.


But in the worst-case scenario, the smallest cut in period $t+1$ is $x^{t+1}$ such that $x_i^{t+1} < x^{t+1} < \hat x_i^{t+1}$, so that the rest of the cake $[x^{t+1},1]$ is worth less than 
$$
u_i([x^{t+1},1];\theta_i) < \frac{n-t-1}{n-t} u_i([x^{t},1];\theta_i)
$$

Following this argument, we see that agent $i$'s only alternative is to receive the last piece of cake $[x^{n-1},1]$, which is, by construction 
$$
u_i([x^{n-1},1];\theta_i) <  u_i([x^{n-2},x_i^{n-1}];\theta_i) < \ldots < u_i([x^{t-1},x_i^{t}];\theta_i) =\frac{u_i([x^{t-1},1];\theta_i)}{n-t+1}
$$

This concludes the proof that no manipulation yields a higher utility than truth-telling in the worst-case scenario, so inequality \eqref{eq:nom-inf} is satisfied.

Next, we show that no manipulation yields a higher utility in the best-case scenario.
We use the following lemma.

\begin{lemma}
At period $t$, 
\begin{equation*}
\sup_{\theta_{-i}} \, u_i(\ML(E(\theta_i,\theta_{-i}));\theta_i) = u_i([x^{t-1},1];\theta_i)
\end{equation*}
This is, agent $i$, by being truthful in period $t$, can expect (in the best-case scenario) to obtain the whole cake available in period $t$.
\end{lemma}
\begin{proof}
 Let us remember that truthful agent $i$ cuts the cake at a point $x_i^t$ such that $u_i([x^{t-1},1];\theta_i)=\frac{u_i([x^{t-1},1];\theta_i)}{n-t+1}$. In the best-case scenario, the smallest cut at period $t$ is $x^t<x_i^t$ such that $x^t=x^{t-1}+\epsilon$. This guarantees that $u_i([x^{t-1},1];\theta_i) + \epsilon' \approx u_i([x^{t},1];\theta_i)$ for an arbitrarily small $\epsilon'$, i.e. no valuable piece of cake for agent $i$ has been allocated. That we can find an $\epsilon$ so small comes directly from the standard assumption that the utilities are divisible.
	
The best-case scenario is that all further smallest cake cuts $x^{t+1},x^{t+2},\ldots$ are epsilon increments of $x^t$, such that in the end agent $i$ receives $X_i=[x^{n-1},1]$, which gives him a utility arbitrarily close to $u_i([x^t,1])$, i.e. the utility of eating all of the remaining cake available. 
\end{proof}
Since this is the maximum utility attainable,
inequality \eqref{eq:nom-sup} in the NOM definition is satisfied.
This concludes the proof that no manipulation gives a higher utility to a truthful agent in the best-case scenario.

We conclude that no manipulation is better than truth-telling in either the best or the worst-case scenario, thus no manipulation is obvious and leftmost leaves is NOM.

\begin{remark}
A minor modification of leftmost leaves, sometimes used in the literature \citep{procaccia2016}, destroys the NOM of the procedure. In this modification, at each period $t$, agents are asked $\text{cut}(i ; x^{t-1},1/n)$.

To see that this variant is obviously manipulable, consider the case of an agent who has uniform preferences over the whole $[0,1]$ interval and has to cut the cake against 4 other agents. Her first cut must be $x^1_i=0.2$, which guarantees a utility of 0.2. However, suppose the lowest cut (submitted by someone else) was instead $x^1=0.1$. In period $t=2$, the remaining cake is $[0.1,1]$ and if agent $i$ is truthful, she should cut the cake at $x_i^2=0.3$ so that $u_i([0.1,0.3];\theta_i)=1/5=0.2$. However, there is a manipulation that is guaranteed to yield a higher utility in the worst-case scenario. If she cuts the cake instead at $\hat x_i^2=0.325$ (i.e. the point at which $u_i([0.1,\hat x_i^2];\theta_i)=1/4$), in the worst-case scenario, in which her cut is the lowest, she would guarantee herself a utility of $0.225 > 0.2$. If her cut was not the lowest, by continuing to cut the cake at the point $\hat x_i^t$ such that $x^t_i=\frac{([x^{t-1},1])}{n-t+1}$ for all subsequent $t$, she could make sure to receive a utility of at least $0.225$ (Lemma \ref{lem:inf}), which is larger than the worst-case scenario utility received by being truthful, 0.2. We conclude that the modified version of leftmost leaves is obviously manipulable, and thus the choice of which query to make at each step of the mechanism is crucial to make the leftmost leaves procedure not-obviously manipulable. 
\end{remark}

\section{Direct-revelation mechanisms}
\label{sec:envyfreeness}
While leftmost leaves is proportional and connected, it does not satisfy other desirable properties such as \emph{envy-freeness} (no agent prefers the piece of cake received by someone else over his own piece) and \emph{Pareto-optimality} (no other allocation is better for one agent and not worse for the others).

In the Robertson-Webb model, we could not yet find NOM mechanisms satisfying these properties. 
For example, the classic mechanism of Selfridge-Conway for three agents (see \citealp{brams1996} for a detailed description)
is envy-free, but it is obviously-manipulable. 
In this mechanism, the first agent cuts the cake into three pieces of equal worth. A truthful agent knows that one of those pieces will never belong to him, and thus he can achieve a maximum utility of at most 0.67. However, a lying agent can cut the cake in one piece of value $1-\epsilon$, and two pieces of almost no value at all. In the best case scenario, he will keep the most valued piece entirely, showing that the Selfridge-Conway procedure is obviously manipulable.

Interestingly, NOM is easier to achieve in the direct-revelation model.

\begin{lemma}
\label{lem:direct-inf}
Every direct-revelation mechanism that always returns proportional allocations satisfies inequality \eqref{eq:nom-inf}.
\end{lemma}
\begin{proof}
By proportionality, a truthful agent always receives a utility of at least $1/n$.
Consider now an untruthful agent $i$ who reports a type  $\theta_i'\neq \theta_i$ (equivalently, reports a utility function $u_i'\neq u_i$).
Consider the case when all other $n-1$ agents have a utility of $u_i$ (the true utility of agent $i$). 
A proportional mechanism must give each of these $n-1$ agents a piece with a value, by the function $u_i$, of at least $1/n$.
Hence, the piece remaining for agent $i$ has a value, by the function $u_i$, of at most $1/n$.
Hence, in inequality \eqref{eq:nom-inf}, the infimum is at most $1/n$ at the left and at least $1/n$ at the right, and the inequality holds.
\end{proof}

\begin{lemma}
\label{lem:direct-sup}
Every direct-revelation mechanism that always returns Pareto-optimal allocations satisfies inequality \eqref{eq:nom-sup}.
\end{lemma}
\begin{proof}
Consider the case when agent $i$ is truthful, all other $n-1$ agents assign a positive value only to a tiny fraction of the cake, and assign a value of $0$ to the rest of the cake.
A Pareto-optimal mechanism must assign almost all the cake to agent $i$. Hence, in inequality \eqref{eq:nom-sup}, the supremum in the right-hand side equals $1$ and the inequality holds.
\end{proof}

\begin{theorem}
\label{thm:nash}
There exists a NOM direct-revelation mechanism that finds envy-free and Pareto-optimal allocations.
\end{theorem}
\begin{proof}
The \emph{Nash-optimal mechanism} is a direct-revelation mechanism that, given $n$ utility functions, selects an allocation that maximizes the product of utilities. Such an allocation is known to be Pareto-optimal and envy-free 
\citep{segal2019monotonicity}, hence it is also proportional.
Hence, by lemmas \ref{lem:direct-inf} and \ref{lem:direct-sup}, the mechanism it is NOM.
\end{proof}
When the utility functions are \emph{picewise-constant}, the Nash optimal mechanism can be computed by an efficient algorithm described by  \cite{aziz2014cake}.

In contrast to the leftmost leaves rule, the Nash-optimal rule may return disconnected pieces.
Moreover, it is known that \emph{any} Pareto-optimal envy-free rule may have to return disconnected pieces 
(see Example 5.1 in \citet{segal2018resource}).
Since Pareto-optimality is crucial in the proof of Theorem \ref{thm:nash},
it remains an open question whether there exists a NOM mechanism that is both envy-free and connected.

A related question is whether there exists an algorithm that finds proportional, Pareto-optimal and connected allocations. If such an algorithm exists, then by Lemmas \ref{lem:direct-inf} and \ref{lem:direct-sup}, it is NOM.

\section{Conclusion and Experimental Evidence}
Although it is impossible to cut a cake in a strategy-proof manner that is not completely unfair to some agent, we can divide a cake in a fair, proportional way that cannot be obviously manipulated using an easily implementable mechanism called leftmost leaves. 

Troyan and Morril's notion of NOM not only allows us to escape the tradeoff between fairness and incentives in cake-cutting, but also helps us to better understand real-life behavior when dividing an heterogeneous good. In the first lab experiment comparing cake-cutting mechanisms, \cite{kyropoulou2019fair} report truthful behavior in NOM cake-cutting mechanisms in 44\% of the cases, whereas the respective number for OM ones is of 31\% (the difference is statistically significant with a p-value of 0.0000). In particular, leftmost leaves was significantly less manipulated than Banach-Knaster last diminisher when agents played against 2 opponents (difference of 29 percentage points, p-value of 0.0000) and 3 opponents (difference of 16 percentage points, p-value of 0.0000). The Even-Paz modification of leftmost leaves (which is NOM) was also significantly less manipulated than Banach-Knaster (difference of 35 percentage points, p-value of 0.0000).

Although in general NOM gives us testable predictions that map relatively well to observed behavior, it is intriguing that the Selfridge-Conway procedure also reports high rates of truth-telling, comparable to those of NOM mechanisms. Explaining this puzzling phenomenon remains an open problem which we leave for future research.


	\section*{References}
	\setlength{\bibsep}{0cm}
	\bibliographystyle{ecta}
	\bibliography{bibliog}
	
\end{document}